\documentclass[%
 reprint,
 amsmath,amssymb,
 aps,
]{revtex4-2}

\usepackage{graphicx}
\usepackage{subcaption} 
\usepackage{dcolumn}
\usepackage{bm}
\usepackage{amsthm}

\newcommand\Tr{\operatorname{Tr}}

\newcommand\bra{\left<}
\newcommand\ket{\right>}
\newcommand{\braket}[1]{\bra#1\ket}
\newcommand{\mf}[1]{\mathfrak{#1}}

\def\al {\alpha}
\def\be {\beta}

\def\DE {\Delta}

\def\ga {\gamma}
\def\GA {\Gamma}

\def\la {\lambda}

\def\om {\omega}

\def\pr {\partial}

\newtheorem{theorem}{Theorem}
\newtheorem{definition}{Definition}

\begin{document}

\preprint{APS/123-QED}

\title{Gaussian approximation and its corrections for driven dissipative Kerr model}

\author{K.Sh. Meretukov}
\email{meretukov.khazret@gmail.com}
 \affiliation{Faculty of Physics, Lomonosov Moscow State University, Leninskie Gory, Moscow 119991, Russia}
\author{A.E. Teretenkov}%
 \email{taemsu@mail.ru}
\affiliation{%
 Department of Mathematical Methods for Quantum Technologies, Steklov Mathematical Institute of Russian Academy of Sciences, ul. Gubkina 8, Moscow 119991, Russia
}%

\date{\today}

\begin{abstract}
We develop a systematic projection-operator technique for constructing Gaussian approximations and their perturbative corrections in bosonic nonlinear models. As a case study, we apply it to the driven dissipative Kerr oscillator. In the absence of external driving, the model can be solved exactly within a low-dimensional Fock subspace, leading to strongly non-Gaussian states. Nevertheless, we demonstrate that the evolution of first- and second-order moments is captured by our Gaussian scheme with high accuracy even in this regime, providing a natural benchmark. For the general case with external driving, our approach reduces the equations of motion to a closed system for means and covariances and allows one to compute systematic corrections beyond the Gaussian level in closed form. We also calculate the dynamics of linear and quadratic combinations of creation and annihilation operators in  both weak- and strong-drive regimes.
\end{abstract}

\maketitle


\section{\label{sec:introduction}Introduction}

    The study of systems with Kerr nonlinearity is an actively discussed problem \cite{KerrIntr, Beaulieu2025DPT, Miao2024Nonreciprocity, Tang2024Multiphoton, Ravets2025, Najafabadi2023Quasiclassical, Zhang2024KerrCombs, Singh25}. The one-mode dissipative Kerr model has several experimental realizations in semiconductor microcavities, quantum circuits and within optomechanical setups \cite{asjad2023joint}. The Kerr model is also used to squeeze light to improve the sensitivity of interferometers \cite{Andrianov2023KerrSqueezing, Andrianov:24, Anashkina:20}.
    
    However, when considering the Kerr model, there is a problem with the solution when the number of particles is large \cite{MaslovN}. One of the solution methods is the solution by means of classical stochasticity, but in this case all quantum effects are lost \cite{MaslovSolve}. However, there are techniques to recover some of them  \cite{corney2015non}. Another way to deal with the large number of particles is to assume approximate Gaussian dynamics. This allows at least Gaussian quantum effects to be taken into account. And it is a natural assumption if the initial states are assumed to be Gaussian \cite{joneckis1993quantum, banerjee1993interaction, genoni2009enhancement, roman2015parametric,  bolandhemmat2023quantum}. So we focus our study on approximately Gaussian dynamics.

    A theoretical advantage of the Gaussian ansatz is that  Gaussian states are well-studied and occur in many physical setups.  Hence, it is natural to approximate general dynamics with a Gaussian one in order to  apply well-developed approaches to it. In particular, the  tomograms  of Gaussian states are  explicitly known in many tomographic setups \cite{MankoManko2003PhotonNumberTomography, MankoManko2009PhotonOptical, MankoSharapovShchukin2001}.
      Additionally, it has been demonstrated in  \cite{mele2024learningquantumstatescontinuous} that Gaussian states are efficient for quantum tomography in general.  Also, the Gaussian ansatz plays the key role for quasiclassical approximation of open quantum systems \cite{graefe2018lindblad}. The Gaussian ansatz naturally arises as the Gibbs state of coupled harmonic oscillators. In particular, it allows one to describe temperature measurements in them \cite{LopezSaldivar2021}.  Recently, there has also been an increase in interest in Gaussian ansatz approximations in other fields \cite{GaussState}. With minor modifications, the method we developed may be useful in these fields as well.

         Our approach considers cases when the dynamics at zero order are described by a Gaussian channel and cases when non-Gaussian channels in real physical systems are considered nonlinear but Gaussian. This peculiarity is useful in quantum information theory because many properties of Gaussian channels are known \cite{Kholevo}.   Our approach can be considered a quantum analog of functional mechanics with a Gaussian ansatz which has proven its usefulness in the classic case \cite{TrushechkinVolovich2009, Volovich2011}.
    
    In Section \ref{sec:GeneralApproch} we develop a general technique for  Gaussian approximation of the dynamics of bosonic modes. We use projection methods with the Kawasaki–Gunton projector \cite{kawasaki1973theory, rau1996reversible, semin2020dynamical}. Recently, it has been formulated for general ansatzes \cite{meretukov2024time}, so it is very natural to apply it to the special case of the Gaussian ansatz. We emphasize that such  a technique allows one not only to describe dynamics under Gaussian approximation, but also to obtain  perturbative corrections to it. We show that the correction can be computed in closed form for the Gaussian ansatz.
    
    In Section \ref{sec:ConsideredSystem} we apply our results to the driven dissipative Kerr model \cite{asjad2023joint} in the leading order of perturbation theory. In this case it essentially reduces to the application of the Wick's theorem for the right-hand side of the averaged Heisenberg equations. In Section \ref{sec:feq0} we compare the results of our approach with exact solution in the space of density matrices supported at the lowest Fock states. Even in such a non-Gaussian case we show, that the results predicted by our approach are very close to the exact ones. In Section \ref{sec:Efc} we discuss the corrections to the first-order Gaussian approximation in the case of a weak external field.

\section{\label{sec:GeneralApproch}Gaussian approximation with systematic corrections}

    In this section we introduce a general scheme, which allows to obtain the equations describing the dynamics under assumption that the density matrix is approximately Gaussian and systematic corrections taking into account that this approximation slightly breaks during the dynamics. We follow the projection approach with the  Kawasaki–Gunton projector  \cite{kawasaki1973theory, rau1996reversible, semin2020dynamical, meretukov2024time} for the general ansatz parameterized by averages of some relevant observables. Here, we specify our results for the case, where the ansatz is Gaussian and the relevant observables are creation and annihilation operators and their products.  But,  instead of using parametrization by their averages, we reparameterize  the Gaussian states in a more convenient way in terms of the covariance matrix along with the vector of means.

    \subsection{Second order equations for Gaussian ansatz}
	
	To formulate our approach explicitly,  let us introduce some compact notation similar to \cite{Review, Teretenkov2017drift}. Let  $\mf{a} = (a_1, a_2, \ldots, a_d, a_1^{\dagger}, a_2^{\dagger},\ldots,a_d^{\dagger})^T$, where $a_j^{\dagger}$ and $a_j$ are bosonic creation and annihilation operators. The Gaussian states can be parametrized by the vectors of means $ m $ of creation and annihilation  operators and a covariance matrix $ C $, so we use the notation $\rho_{ans}(m, C)$ for such a state. In our notation
    \begin{align}
        \operatorname{Tr} \mathfrak{a} \rho_{ans}(m, C) &= m, \label{eq:mDef} \\
        \operatorname{Tr} (\mathfrak{a} - m)(\mathfrak{a} - m)^T \rho_{ans}(m, C) &= C - \frac{J}{2}, \label{eq:CDef} 
    \end{align}
    where trace is taken element-wisely, e.g. the first equality is equivalent to $\operatorname{Tr} \mathfrak{a}_j \rho_{ans}(m, C) = m_j$, and
    \begin{equation}
    J=
     \begin{pmatrix}
           0 & -I_d \\
I_d & 0
        \end{pmatrix}
    \end{equation}
	is a matrix of the symplectic form.

    For the Gaussian ansatz, the Kawasaki–Gunton projector maps an arbitrary density matrix to a Gaussian state. That is, we consider such operators $\mathcal{P}$ that
	\begin{eqnarray}
		\mathcal{P}^2(t) &=& \mathcal{P}(t), \\
		\mathcal{P}(t)\rho_I(t) &=& \rho_{ans}(m(t), C(t)),
    \end{eqnarray}
    where $\rho_I(t)$ is a density matrix in the interaction picture.
    
    Using such a projector and assuming that $\rho_I(t)$ satisfies the equation
    \begin{equation}
		\label{eq:MainEq}
		\dfrac{d}{dt}\rho_I(t) = \la\mathcal{L}(t)\rho_I(t), 
	\end{equation}
    and at the initial time moment $t_0$ the density matrix $\rho_I(t)$ is Gaussian
    \begin{equation}
        \rho_I(t_0) = \rho_{ans}(m(t_0), C(t_0))
    \end{equation}
    we obtain that in the weak coupling limit $\la \rightarrow 0$ the following system for $m(t)$ and $C(t)$ is fulfilled (see App.~\ref{sec:SecOrdEq} for the derivation):
    \begin{widetext}
        \begin{align}
          \frac{d}{dt} m =& \lambda   \langle \mathcal{L}^*(t)  \mathfrak{a}  \rangle + \lambda^2 \left(  \int_{t_0}^t dt_1 \langle  \mathcal{L}^*(t_1)  \mathcal{L}^*(t)\mathfrak{a}  \rangle  	-   \Tr\left( \mathfrak{a} R_t(m,C) \int_{t_0}^t dt_1 R_{t_1}(m,C) \right) \right), \label{eq:secOrdEqm}\\
    \frac{d}{dt} C =& \lambda   \langle \mathcal{L}^*(t)  \mathfrak{A}(m)   \rangle + \lambda^2 \left(  \int_{t_0}^t dt_1 \langle  \mathcal{L}^*(t_1)  \mathcal{L}^*(t) \mathfrak{A}(m)  \rangle  	-   \Tr\left( \mathfrak{A}(m)  R_t(m,C) \int_{t_0}^t dt_1 R_{t_1}(m,C) \right) \right), \label{eq:secOrdEqC}
    	\end{align}
    where $\mathfrak{A}(m) := \mathfrak{a} \mathfrak{a}^T -m \mathfrak{a}^T - \mathfrak{a} m^T$, $\langle \; \cdot \; \rangle : = \Tr ( \; \cdot \; \rho_{ans} (m, C) )$,
    \begin{equation}\label{eq:RtDef}
        R_t(m,C) :=  \langle \mathcal{L}^*(t)  \mathfrak{a}  \rangle^T 
        \left( \frac{\partial}{\partial m} \rho_{ans}(m, C ) - 2 \left(\frac{\partial}{\partial C} \rho_{ans}(m, C) \right)m \right)
        + \sum_{ij} \langle \mathcal{L}^*(t)\mathfrak{a}_i \mathfrak{a}_j  \rangle \frac{\partial}{\partial C_{ij}} \rho_{ans}(m, C) .
    \end{equation}
    \end{widetext}

    The terms $O(\lambda^3)$ are omitted here, so this is an analog for the second order time-convolutionless master equations in the open quantum system theory \cite[Sec. 9.3]{BreuerPetruccione2002}.  The further corrections can be obtained in the similarly but further cumulant expansion \cite{meretukov2024time,KarasevTeretenkov2023} with Kawasaki–Gunton projector, but is natural to assume that the most important features can be captured by the first two orders of the perturbation theory. 
    
    Remark that we do not use here the Bogolubov-van Hove limit \cite{bogoliubov1946problems, van1955energy, teretenkov2021non}, in particular no analog of the secular approximation is applied here which is essential for obtaining Gorini–Kossakowski–Sudarshan–Lindblad (GKSL) form of usual weak coupling limit master equations, but reduce the accuracy of equations in the initial stages of dynamics. In this work, we want to keep this accuracy at the initial stages. And so we do not consider analogs of secular approximation in this work, but let us remark that such approximation can be done based on algebraic perturbation theory \cite{bogaevski2012algebraic} similarly to Refs. \cite{basharov2020global, trubilko2020hierarchy, basharov2021effective, trubilko2021effective, basharov2024atom}. Such approximation could be important to obtain the equations that preserve the positivity of the matrix $(C - \frac{1}{2} J)E$, which follows from the Robertson–Schrödinger uncertainty relations and canonical commutation relations.

   Eqs.~\eqref{eq:secOrdEqm}-\eqref{eq:secOrdEqC} can be considered as non-linear analog of the equations for the vector of means and covariance matrix for the quantum dynamical semigroups with  quadratic generators \cite{DodonovManko1988, DodonovManko1985, HeinosaariHolevoWolf2010, Review} and differential parametric formalism for the Gaussian states evolution \cite{LopezSaldivar2020}. Remark that Eqs.~\eqref{eq:secOrdEqm}-\eqref{eq:secOrdEqC} are generally non-linear already in terms the vector of means and covariance matrix in contrast to other parameterizations of Gaussian ansatz which can be non-linear even for quadratic generators  \cite{Teretenkov2016, LopezSaldivarManoko2021}. Their solution defines the transformation $(m(0), C(0)) \mapsto (m(t), C(t)) $, which can be considered as a nonlinear analog of the transformation, which defines the quantum Gaussian channels \cite[Chapter 12]{Kholevo}.  

    Now let us discuss how to  the RHS of Eqs.~\eqref{eq:secOrdEqm}--\eqref{eq:secOrdEqC} explicitly. Typically,  to transform a master equation into the form defined by  Eq.~\eqref{eq:MainEq} one needs to move to the interaction picture. If one calculates $\mathcal{L}(t)$ in terms of creation and annihilation operators, then $  \mathcal{L}^*(t)  \mathfrak{a}$ and $  \mathcal{L}^*(t)  \mathfrak{A}(m)$ are also expressed in terms of them, hence, terms of the form $\langle \mathcal{L}^*(t)  \mathfrak{a} \rangle$ and  $\langle \mathcal{L}^*(t)  \mathfrak{A}(m) \rangle$ can be calculated by the Wick's theorem \cite{NosalTeretenkov}. So we discuss the calculation of  $  \mathcal{L}(t)  \mathfrak{a}$ in terms of creation and annihilation operators in Subsec.~\ref{subsec:interaction}.

    But one still need to calculate the terms $ \Tr ( \mathfrak{a}  R_t(m,C)  R_{t_1}(m,C) )$ and $ \Tr ( \mathfrak{A}(m)  R_t(m,C)   R_{t_1}(m,C) )$. In Subsec.~\ref{subsect:derGaussStates} we show that it can be done in the similar manner, because both derivatives of the Gaussian states in Eq.~\eqref{eq:RtDef} can be reduced to the product of Gaussian states and creation and annihilaiton operators, and similar thing can be done for the products products $ R_t(m,C)  R_{t_1}(m,C) $.

      \subsection{Interaction picture with dissipative free dynamics}\label{subsec:interaction}

    Eq.~\eqref{eq:MainEq} is written in the interaction representation, while usually one starts with  the master equation 
    \begin{equation}
		\label{eq:EvEq}
		\dfrac{d}{dt}\rho(t) = (\mathcal{L}_0 + \lambda\mathcal{L}_I)\rho(t),
	\end{equation}
    where $\mathcal{L}_0 $ and $\mathcal{L}_I $ have the GKSL form. So we need to move into the interaction representation \eqref{eq:MainEq}. The density matrix is transformed according to the following equation 
	\begin{equation}
		\label{eq:IntRepRho}
		\rho_I(t) := e^{-\mathcal{L}_0t}\rho(t),
	\end{equation}
	then the density matrix \eqref{eq:IntRepRho} satisfies Eq.~\eqref{eq:MainEq} with
	\begin{equation}
		\label{eq:IntRepL}
		\mathcal{L}(t) = e^{-\mathcal{L}_0t}\mathcal{L}_Ie^{\mathcal{L}_0t}. 
	\end{equation}

    Remark that $\mathcal{L}(t)$ can be not of GKSL form, even if $\mathcal{L}_0 $ and $\mathcal{L}_I $ do. In particular, if $\mathcal{L}_0 $  has non-trivial dissipator and do not commutes with  $\mathcal{L}_I $, i.e. in the generic case, then  $\mathcal{L}(t)$ is not of GKSL form.

 As we consider the dynamics that is approximately Gaussian then we should assume that the dynamics with a free generator $\mathcal{L}_0 $ preserves the Gaussian states, which occurs if and only if $\mathcal{L}_0 $ is a GKSL generator which is quadratic in creation and annihilation operators \cite{HeinosaariHolevoWolf2010}. So we focus on such quadratic free generators $\mathcal{L}_0 $ in our work. In our notation it takes the form \cite{Review}
 \begin{align}
\mathcal{L}_0(\rho)=&-i\left[\frac{1}{2} \mathfrak{a}^T H \mathfrak{a}
+f^T \mathfrak{a}, \rho\right] \nonumber\\
&+\mathfrak{a}^T \rho \Gamma \mathfrak{a}
- \left\{\frac{1}{2} \mathfrak{a}^T \Gamma^T \mathfrak{a}, \rho \right\},
 \end{align}
 where $H$ and $\Gamma$ are $2d \times 2d$ matrices and $f$ is $2d$-dimensional vector.
 
  For such free generators one can  obtain ''the Heisenberg equations'' for superoperators of left and right multiplication by creation and annihilation operators \cite{DodonovManko1988, DodonovManko2003, ProsenSeligman2010, Dis, KarasevTeretenkov2025, gaidash2025lindblad} 
    \begin{equation}
		\mathfrak{B} :=
        \begin{pmatrix}
            \mathfrak{a}\cdot\\
            \cdot E\mathfrak{a}
        \end{pmatrix}, \qquad 	E := \begin{pmatrix}
			0 & I_d \\
			I_d & 0
		\end{pmatrix},
	\end{equation}
    which is a variant of the vectorization \cite{Baranger1958, Havel2003}. Namely, $\mathfrak{B}(t) := e^{-\mathcal{L}_0t} \mathfrak{B} e^{\mathcal{L}_0t}$ satisfies \cite[Subsec.~1.6.5]{Dis}
	\begin{equation}
		\label{eq:InteractionReprSuper}
		\frac{d}{dt}\mathfrak{B}(t) = -J_2L \mathfrak{B}(t) - J_2 F, 
	\end{equation}
	where
    \begin{equation}
L := 
\begin{pmatrix}
-i H-\frac{\Gamma^T+\Gamma}{2} & \Gamma E \\
E \Gamma^T & i E H E-E \frac{\Gamma^T+\Gamma}{2} E
\end{pmatrix},
    \end{equation}
    \begin{equation}
     J_2 := 
\begin{pmatrix}
- J & 0\\
0 & -J
\end{pmatrix},
\qquad
F :=
\begin{pmatrix}
-i f \\
i E f
\end{pmatrix}.
    \end{equation}

	Since in our case all coefficients do not depend on time, it follows that the solution of this equation is the function
	\begin{equation}
		\label{eq:SolOfSuper}
		\mathfrak{B}(t) = e^{- J_2 L t} \mathfrak{B}(0) +  \frac{ e^{- J_2 L t}  - 1}{J_2 L} J_2 F.
	\end{equation}

	Then to calculate the interaction representation $\mathcal{L}(t) $ for any superoperator $\mathcal{L}_I$, which is a polynomial of left and right multiplication by creation and annihilation operators, one can just multiply the correspondent elements of the vector $\mathfrak{B}(t)$. For example, if one wants to find $[a^{\dagger}_1a_1^{\dagger}a_1a_1, \; \cdot \;](t)$ we need to calculate the expression
	
	\begin{equation}
    \label{eq:intrepsol}
		\mathfrak{B}_2(t)\mathfrak{B}_2(t)\mathfrak{B}_1(t)\mathfrak{B}_1(t) - \mathfrak{B}_4(t)\mathfrak{B}_4(t)\mathfrak{B}_3(t)\mathfrak{B}_3(t),
	\end{equation}
	where $\mathfrak{B}_j(t)$ is the $j$-th element of the vector $\mathfrak{B}(t)$. Similarly the more complex interaction generators, which are polynomials  of creation and annihilation operators (right- or left-multiplied), can be transformed into the interaction picture if the free generators are quadratic in creation and annihilation operators. 

     \subsection{Derivatives and squares of Gaussian states}
    \label{subsect:derGaussStates}
     
	Eq.~\eqref{eq:RtDef} contains the ansatz derivatives
	\begin{equation}
		 \frac{\partial }{\partial x} \rho_{ans}(m, C),
	\end{equation}
    where $x$ is a one of the components $m_j$ or $C_{ij}$.

	The formula of the ansatz derivative obtained in \cite{Review} can be represented in the following form
	
	\begin{equation}
		\label{eq:DerWithLin}
		\frac{\partial}{\partial x}\rho_{ans}(m, C) = \left(\frac{1}{2}  \mf{a}^TM\mf{a} + \mf{a}^TG + c\right)\rho_{ans}(m, C),
	\end{equation}
	where
    \begin{eqnarray}
        M &=& \frac{I}{C - \frac{J}{2}}\left(\frac{\partial}{\partial x} C\right)\frac{I}{C + \frac{J}{2}},\\
        G &=&\frac{I}{C - \frac{J}{2}} \left( C + \frac{J}{2} \right)\frac{\partial}{\partial x} \left(\frac{I}{C + \frac{J}{2}}m\right),\\
        c &=& -\frac{m^T}{2}\frac{I}{C - \frac{J}{2}}J \frac{\partial}{\partial x} \left(\frac{I}{C + \frac{J}{2}}m \right) \nonumber\\
            & & - \frac{1}{2} \frac{\partial}{\partial x} \left(m^T\left( \frac{1}{2}\frac{I}{C + \frac{J}{2}} + \frac{1}{2}\frac{I}{C - \frac{J}{2}}\right.\right. \nonumber \\
            & & +\left.\left. 2\text{arcoth}(-2J^{-1}C)J^{-1} \right)m + s\right),
    \end{eqnarray}
    with 
    \begin{eqnarray}
        s &=& \frac{1}{2}\ln\left(\left|\det\left( -\frac{I}{C-\frac{J}{2}}J \right)\right|\right) \nonumber \\& &+ m^T\text{arccoth}^T(-2J^{-1}C)m. \label{eq:sInC}
    \end{eqnarray}
    A detailed derivation is provided in  Appendix \ref{App}.
	
	However, in fact in Eq.~\eqref{eq:secOrdEqC} we have this expression squared, in order to get the desired degeneracy the density matrix has to be carried through the quadratic form, i.e.

    \begin{align}
		\rho_{ans}(m, C) & \left(\frac{1}{2}\mf{a}^TM\mf{a} + f^T\mf{a} + c\right) \nonumber\\
        &= \left(\mf{a}^TM'\mf{a} + f'^T\mf{a} + c'\right)\rho_{ans}(m, C), \label{eq:DenMatrOverForm}
	\end{align}
	where
	
	\begin{align}
		M' &= \frac{1}{2}\frac{I}{C-\frac{J}{2}}\left(C + \frac{J}{2}\right)MJ\frac{I}{C+\frac{J}{2}}\left(C - \frac{J}{2}\right)J^{-1}, \\
		f' &= \frac{I}{C-\frac{J}{2}}\left(C + \frac{J}{2}\right)MJ\frac{I}{C+\frac{J}{2}}m \nonumber\\&+ \frac{I}{C - \frac{J}{2}} \left( C + \frac{J}{2} \right)f, \\
		c' &= 2\left(f^T - m^T\frac{1}{2}\frac{I}{C-\frac{J}{2}}JM\right)\frac{J}{2}\frac{I}{C+\frac{J}{2}}m + c. \label{eq:cPrimec}
	\end{align}
	
	We also need to find a renormalization for the ansatz square. 
	
	\begin{equation}
		\label{eq:RhoSq}
		(\rho_{ans}(m, C))^2 = \frac{1}{\sqrt{|\det(2 C)|}} \rho_{ans}(m, C'),
	\end{equation}
	where 
	\begin{equation}
		C' = - \frac{J}{4}\left(\frac{1}{2 J C} + 2 J C\right).
	\end{equation}
	
	The derivation of this relation can be seen in Appendix \ref{app:anzoverform}.
	
	Thus, we can   explicitly calculate the right-hand side of Eq.~\eqref{eq:secOrderEqE} for the Gaussian ansatz.

\section{\label{sec:ConsideredSystem}  Driven dissipative Kerr model}

Let us apply our general scheme to the one-mode Kerr model in the classical external field with dissipation. Namely, we consider  master equation \eqref{eq:EvEq}, where  (see, e.g., \cite{asjad2023joint}) 
	\begin{equation}
		\label{eq:FormOfL}
		(\mathcal{L}_0 + \lambda\mathcal{L}_I)\rho = -i[H_{\lambda},\rho] + \dfrac{\ga}{2}(a\rho a^{\dagger} - \dfrac12\{a^{\dagger}a,\rho\})
	\end{equation}
	and
	\begin{equation}
		\label{eq:FormOfH}
		H_{\la} = -\DE a^{\dagger}a + \dfrac{\la\chi}{2}a^{\dagger}a^{\dagger}aa - iF(a - a^{\dagger}),
	\end{equation}
    where $\Delta$ is the pump-cavity detuning, $F$ is the coherent drive strength, $ \lambda \chi$ is the Kerr anharmonicity ($\lambda$ is introduced to separate the small parameter explicitly), $\gamma$ is the one photon decay rate.

     \subsection{Transformation into interaction representation}
	From Eq.~\eqref{eq:SolOfSuper} it follows that the bosonic operators have the following form in the interaction representation 

    \begin{eqnarray}
        \label{eq:bosonicoperators}
        (a\cdot)(t) &=& \alpha_1(t)(a\cdot) + \beta_1(t), \nonumber\\
        (a^{\dagger}\cdot)(t) &=& \alpha_{21}(t) (a^{\dagger}\cdot) + \alpha_{22}(t)(\cdot a^{\dagger}) + \beta_2(t), \nonumber\\
        (\cdot a^{\dagger})(t) &=& \alpha_{3}(t)(\cdot a^{\dagger}) + \beta_3(t), \nonumber\\
        (\cdot a)(t) &=& \alpha_{41}(t)(\cdot a) + \alpha_{42}(t)(a\cdot) + \beta_4(t),
    \end{eqnarray}
    where 	
    \begin{eqnarray*}
        \alpha_1(t) &=& e^{-\frac{\gamma  t}{4}+i \Delta  t}, \;
        \beta_1(t) = \frac{F \left(4-4 e^{-\frac{\gamma  t}{4}+i \Delta  t}\right)}{\gamma -4 i \Delta }, \\
        \alpha_{21}(t) &=& e^{\frac{1}{4} t (\gamma -4 i \Delta )},  \;
        \alpha_{22}(t) = -2 e^{-i \Delta  t} \sinh \left(\frac{\gamma  t}{4}\right),\\
        \beta_2(t) &=& \frac{F \left(4-4 e^{-\frac{1}{4} t (\gamma +4 i \Delta )}\right)}{\gamma +4 i \Delta },  \;
        \alpha_3(t) = e^{-\frac{1}{4} t (\gamma +4 i \Delta )}, \\
        \beta_3(t) &=& \frac{F \left(4-4 e^{-\frac{1}{4} t (\gamma +4 i \Delta )}\right)}{\gamma +4 i \Delta },  \;
        \alpha_{41}(t) = e^{\frac{1}{4} t (\gamma +4 i \Delta )},\\
        \alpha_{42}(t) &=& -2 e^{i \Delta  t} \sinh \left(\frac{\gamma  t}{4}\right),  \;
        \beta_4(t) = \frac{F \left(1- e^{-\frac{\gamma  t}{4}+i \Delta  t}\right)}{\frac{\gamma}{4} - i \Delta }.
    \end{eqnarray*}
	
	In order to show the necessity of introducing the projector, let us rewrite this equation in the Heisenberg representation.	The generator in the Heisenberg representation $\mathcal{L}^*: \Tr(A\mathcal{L}\rho) = \Tr(\rho\mathcal{L}^*A)$ has the following form (see Appendix \ref{Heiz} for the proof)
	\begin{equation}
		\label{eq:Lst}
		\mathcal{L}^* = i[H_{\la},\cdot]  + \dfrac{\ga}{2}(a^{\dagger}\cdot a - \dfrac{1}{2}\{a^{\dagger}a, \; \cdot \;\}).
	\end{equation}
    
	Since the operators $a, a^{\dagger}$ do not depend on time in the Schrödinger representation, then
	\begin{equation}
		\label{eq:DinOfAv}
		\frac{d}{dt}\braket{a} = \Tr\left(a \frac{d}{dt}\rho\right) = \Tr(a\mathcal{L}\rho) = \Tr(\rho\mathcal{L}^*a).
	\end{equation}
	The Heisenberg equation for the averages $a, a^{\dagger}, a^2, {a^{\dagger}}^2, a^{\dagger}a$ has the form
    \begin{widetext}
	\begin{eqnarray}
			\frac{d}{dt} \braket{a} &=& \braket{a}(-\dfrac{\ga}{4} + i\DE) + F - i\la\chi\braket{Na}, \nonumber \\
			\frac{d}{dt} \braket{a^{\dagger}} &=& \braket{a^{\dagger}}(-\dfrac{\ga}{4} - i\DE) + F + i\la\chi\braket{a^{\dagger}N}, \nonumber \\
			\frac{d}{dt} \braket{N} &=& -\dfrac{\ga}{2}\braket{N} + F(\braket{a} + \braket{a^{\dagger}}), \label{eq:FinDinOfAv} \\
			\frac{d}{dt} \braket{a^2} &=& \braket{a^2}[-\dfrac{\ga}{2} + i(2\DE - \la\chi)] + 2\braket{a}F - 2i\la\chi\braket{Na^2},  \nonumber \\
			\frac{d}{dt} \braket{{a^{\dagger}}^2} &=& \braket{{a^{\dagger}}^2}[-\dfrac{\ga}{2} + i(\la\chi - 2\DE)] + 2\braket{a^{\dagger}}F + 2i\la\chi\braket{{a^{\dagger}}^2N}. \nonumber
	\end{eqnarray}
    \end{widetext}

    As can be seen, this system of equations is not closed. In order to close it, the method of projection operators is used in this paper.

\section{Benchmarking Gaussian approximation by Fock states dynamics\label{sec:feq0}}
	In the following sections we apply our method to the Kerr model in different regimes. At first we consider the regime without external field $F = 0$. Its features are that in the system in this regime there are no mechanisms for increasing the number of particles, which allows us to restrict the system to the $N$-level and solve it analytically, and also, the obtained Fock states are strongly non-Gaussian. Since our method is based on approximation by a Gaussian ansatz, it is important to consider the system in such a non-Gaussian state as worst case benchmark.

            \begin{figure}
    \centering
     \includegraphics[width=\linewidth]{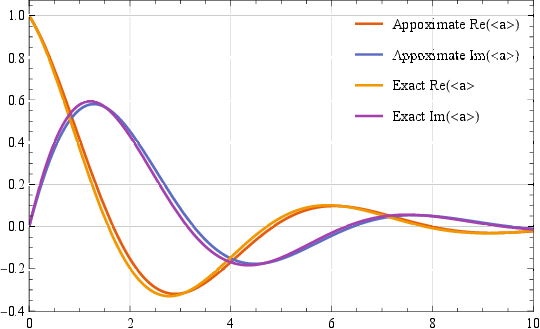}
    \caption{Comparison of our approximation in the first order with exact results for Fock state with $\langle N(0) \rangle=6$ in the absence of the external field (for lower $\langle N(0) \rangle$ is the difference between exact and approximate solutions is even smaller).}
    \label{fig:wick6}
\end{figure}

    The Liouville operator in the interaction representation in this case has the following form
	
	\begin{align}
		\mathcal{L}(t) =& -i\frac{\chi}{2}({a^{\dagger}}^2a^2\cdot - \cdot{a^{\dagger}}^2a^2  \nonumber\\
        &+2   (e^{-\frac{\ga}{2}t} - 1) (a^{\dagger}a^2\cdot a^{\dagger} - a\cdot {a^{\dagger}}^2a) ),\label{eq:InterRep}
	\end{align}
	hence, the conjugate operator has the form
	\begin{align}
		\mathcal{L}^*(t)=& -i\frac{\chi}{2}(\cdot{a^{\dagger}}^2a^2 - {a^{\dagger}}^2a^2\cdot  \nonumber\\
        &+ 2   (e^{-\frac{\ga}{2}t} - 1) ( a^{\dagger}\cdot a^{\dagger}a^2 - {a^{\dagger}}^2a\cdot a) ).\label{eq:InterRepCon}
	\end{align}

	Since we are working with creation/annihilation operators, we need to write down how this operator acts on a arbitrary combination of the form ${a^{\dagger}}^na^m$

    \begin{widetext}
    	\begin{equation}
    		\label{eq:InterRepOnForm}
    		\mathcal{L}^*(t){a^{\dagger}}^na^m = -i\frac{\chi}{2}((m^2 - m - n^2 + n){a^{\dagger}}^na^m + 2 (m - n) e^{-\frac{\ga}{2}t} {a^{\dagger}}^{n + 1} a^{m + 1}).
    	\end{equation}
    \end{widetext}
	
	Since we consider the system without an external field, then this system can be restricted to a two-level system. Then in expressions where annihilation-creation operators occur, the degrees above the second will be nullified. And only the particle number operator acts as a set of relevant observables $P_m$. The action of the operator $\mathcal{L}^*$ on $N$ can be obtained from  Eq.~\eqref{eq:InterRepOnForm} taking $n = m = 1$. Then $	\mathcal{L}^*N = 0.$

	As can be seen, it turns out that the average number of particles is conserved. This result is expected, since in the absence of external influence in the system there is no mechanism for changing the number of particles. However, the result of our approach can be compared with the analytical one. In order to do this, it is necessary to translate our results obtained in the interaction representation into the Schrödinger representation. This can be done using the results of \cite{Review}.

	By doing this, we obtain that our result in the Schrödinger representation is a decaying exponential
		$
		\braket{N(t)} = e^{-\frac{\ga}{2}t}\braket{N(0)}.
	$
	
	Solving this system analytically, we get the same result. It follows that our result agrees with the analytical result.
	
	\begin{equation*}
		\braket{N(t)}_{exact } = \braket{N(t)} = e^{-\frac{\ga}{2}t}\braket{N(0)}.
	\end{equation*}

    The same can be said for the average of creation and annihilation operators.

	We can also compare these results for the case when the higher order moments in Eq.~\eqref{eq:FinDinOfAv} are written out via Wick's theorem. This is exactly Eqs.~\eqref{eq:secOrdEqm}--\eqref{eq:secOrdEqC} approximated in the first order in $\lambda$. As it can be seen on  Fig.~\ref{fig:wick6} even in such non-Gaussian case this equations give good approximation of exact dynamics at least for several lower Fock states, for which the system is quantum and not approximated by classical dynamics. Thus, despite formally our approach is based on Gaussian approximation it is also valid for other practical cases, when the system is in quantum regime.  For example, it is applicable, when the state can be fitted by convex combination of Gaussian state and lower Fock states.

\section{Weak external field\label{sec:Efc}}
	
	In this section we consider the case when an external field is present in the system, i.e. $F \not= 0$. However, considering the general case is a difficult task due to the large number of summands in the generator, so we consider the case where the external field is small and the summands with $F^{1 + n}$, $n > 0$ can be neglected. 

    \begin{figure}[h]
    \centering
     \includegraphics[width=\linewidth]{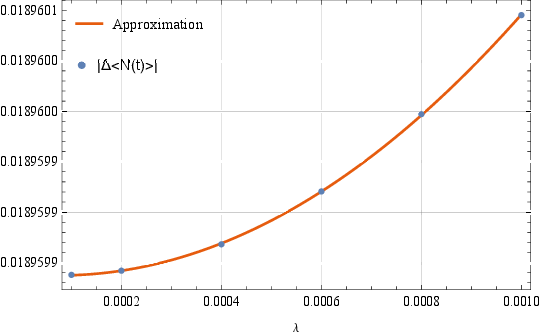}
\caption{Test that the second order  approximation indeed lead to the second order perturbation to its solutions.}
\label{fig:precisionOfWeakEq}
     \end{figure} 

        \begin{figure*} 
    \centering
    \begin{subfigure}[b]{0.32\linewidth}
        \includegraphics[width=\linewidth]{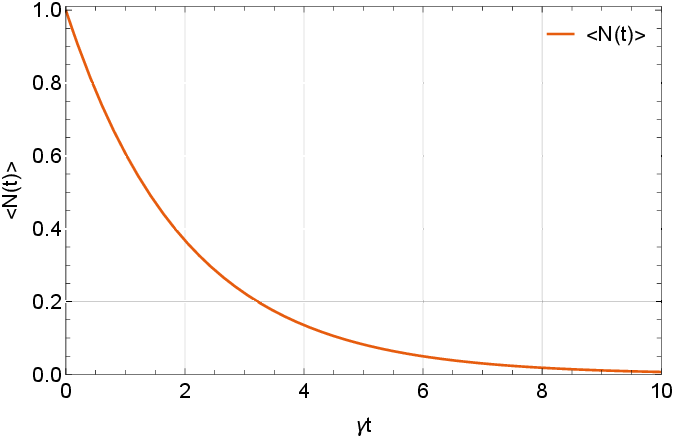}
        \caption{Average number of particles \newline \;\;\;\;\;\;\;\;}
        \label{fig:weakforcen}
    \end{subfigure}\hfill
    \begin{subfigure}[b]{0.32\linewidth}
        \includegraphics[width=\linewidth]{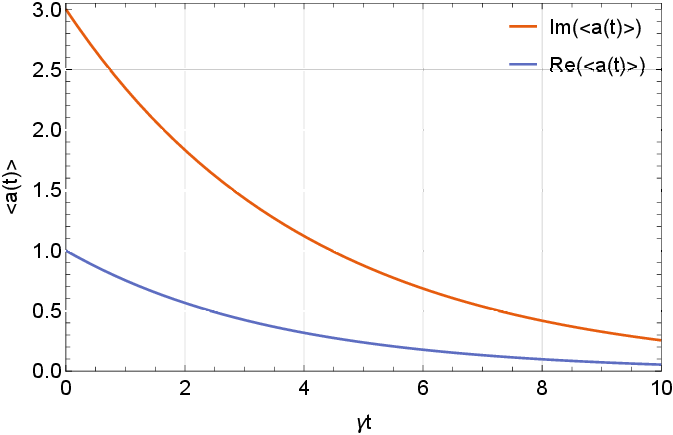}
        \caption{Average of the creation and annihilation operators}
        \label{fig:weakforcea}
    \end{subfigure}\hfill
    \begin{subfigure}[b]{0.32\linewidth}
        \includegraphics[width=\linewidth]{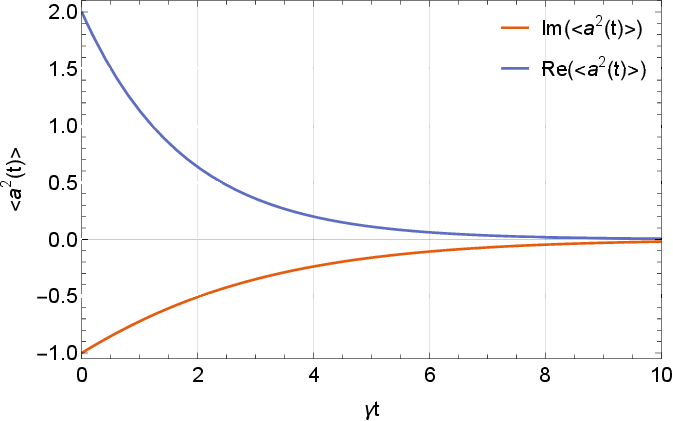}
        \caption{Average of squared creation and annihilation operators}
        \label{fig:weakforceaa}
    \end{subfigure}
    \caption{Dynamics of first and second moments of creation and annihilation operators  in the weak.  Fig. \ref{fig:weakforcea} and Fig. \ref{fig:weakforceaa} were plotted with an additional multiplier $e^{-i\DE t}$ and $e^{-2i\Delta t}$ respectively, i.e. in the rotating frame (This removes frequent oscillations and makes the graph more readable). The weak filed does not affect the monotonicity of relaxation of the moments to their vacuum values occurring without filed.}
    \label{fig:weakforce}
\end{figure*}

	Substituting the results of Eq.~\eqref{eq:bosonicoperators} into Eq.~\eqref{eq:intrepsol} and reducing the necessary terms, we obtain the action of the conjugate generator on the bosonic operator of the form ${a^{\dagger}}^na^m$
    	\begin{eqnarray}
    		\label{eq:ActOfGenOnExpr}
    		\!\!\!\!\mathcal{L}^*(t){a^{\dagger}}^na^m &=& -i\frac{\chi}{2}(c_{00}(t){a^{\dagger}}^na^m + c_{11}(t){a^{\dagger}}^{n+1}a^{m+1} \nonumber\\
            &+& c_{01}(t){a^{\dagger}}^na^{m + 1} + c_{12}(t){a^{\dagger}}^{n + 1}a^{m + 2} \nonumber\\
            &+& c_{0-1}(t){a^{\dagger}}^na^{m-1} + c_{-10}(t){a^{\dagger}}^{n - 1}a^m \nonumber\\&+& c_{10}(t){a^{\dagger}}^{n+1}a^m 
            + c_{21}(t){a^{\dagger}}^{n+2}a^{m+1}),
    	\end{eqnarray}
	where
    	\begin{eqnarray}
    		c_{00}(t) &=& m (m-1)-n (n-1), \\
    		c_{11}(t) &=& 2 (m-n) e^{-\frac{\gamma}{2} t}, \\
    		c_{01}(t) &=& 2 m (\alpha_1(t))^2 \alpha_{21}(t) \beta_2(t) \nonumber \\
            & & -4 n \alpha_3(t) \alpha_{41}(t) \beta_3(t)  (\alpha_{41}(t)+ \alpha_{42}(t)),\\
    		c_{12}(t) &=& 2 (\alpha_1(t))^2 (\alpha_{21}(t) + \alpha_{22}(t) )\beta_2(t) \nonumber \\
            & &  - 2\alpha_3(t) (\alpha_{42}(t))^2 \beta_3(t) \nonumber \\
            & & -2 \alpha_3(t) \alpha_{41}(t) \beta_3(t) (\alpha_{41}(t)+2 \alpha_{42}(t)),
            \end{eqnarray}
    	\begin{eqnarray}
    		c_{0-1}(t) &=& 2 m (m-1) \alpha_1(t) (\alpha_{21}(t))^2 \beta_1(t) ,\\
    		c_{-10}(t) &=&-2 n (n-1) \alpha_3(t) (\alpha_{41}(t))^2 \beta_3(t) ,\\
    		c_{10}(t) &=& 4 m \alpha_1(t) \alpha_{21}(t) \beta_1(t) (\alpha_{21}(t) + \alpha_{22}(t)) \nonumber \\
            & &-2 n(\alpha_3(t))^2 \alpha_{41}(t) \beta_4(t) ,\\
    		c_{21}(t) &=& 2 \alpha_1(t) \alpha_{21}(t) \beta_1(t) (\alpha_{21}(t)+2 \alpha_{22}(t)) \nonumber \\
            & &+2 \left(\alpha_1(t) (\alpha_{22}(t)^2 {\beta_1(t)}-(\alpha_3(t))^2 \alpha_{42}(t) \beta_4(t)\right) \nonumber\\
            & &-2 (\alpha_3(t))^2 \alpha_{41}(t) \beta_4(t).
    	\end{eqnarray} 
        
	As we see, in this case such operator does not give an identical zero when acting on the particle number operator.

   Firstly, we have made additional benchmark on precision of our solutions in this regime. Due to the fact that our equations are highly non-linear, it is natural to ask if the second order contribution in equations leads to the second order contribution in its solutions. The numerical test for this is presented  on Fig.~\ref{fig:precisionOfWeakEq}.
	
	Secondly, in this regime, we use data from \cite{Andrianov:24}. However, the parameters describing our system are $\DE, \ga, \chi, F$. They are natural in theoretical description, but in practice other parameters are more often used. These parameters are: second order dispersion coefficient $\be_2$, nonlinear coefficient $\ga$, linear loss coefficient $\al$, and laser characteristics such as: shape, power and time. Thus, in order to obtain dynamics for realistic regimes, we need to express our parameters through experimentally calculated parameters. A detailed derivation is provided in the appendix \ref{Parrams}. The dynamics of first ans decond moments of of creation and annihilation in this case is presented on Fig.~\ref{fig:weakforce}

	\begin{figure*}[t]
    \centering
    \begin{subfigure}[b]{0.32\linewidth}
        \includegraphics[width=\linewidth]{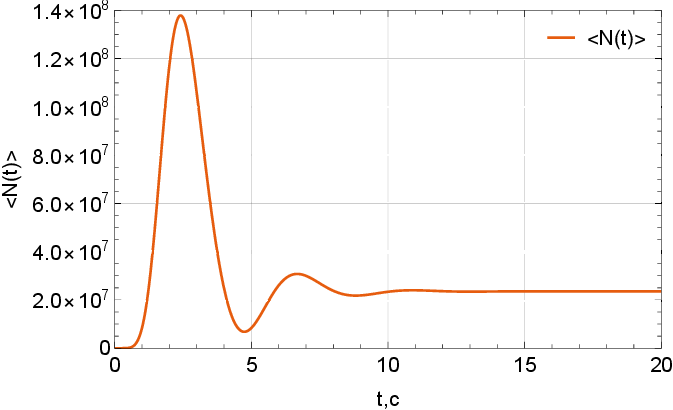}
        \caption{Average number of particles \newline \;\;\;\;\;\;\;\;}
        \label{fig:strongn}
    \end{subfigure}\hfill
    \begin{subfigure}[b]{0.32\linewidth}
        \includegraphics[width=\linewidth]{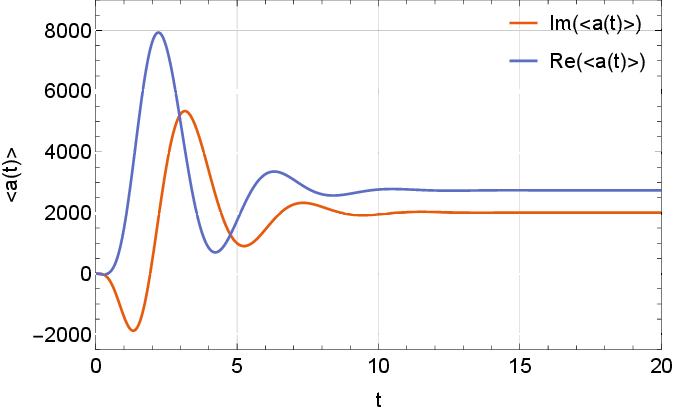}
        \caption{Average of the creation and annihilation operators}
        \label{fig:stronga}
    \end{subfigure}\hfill
    \begin{subfigure}[b]{0.32\linewidth}
        \includegraphics[width=\linewidth]{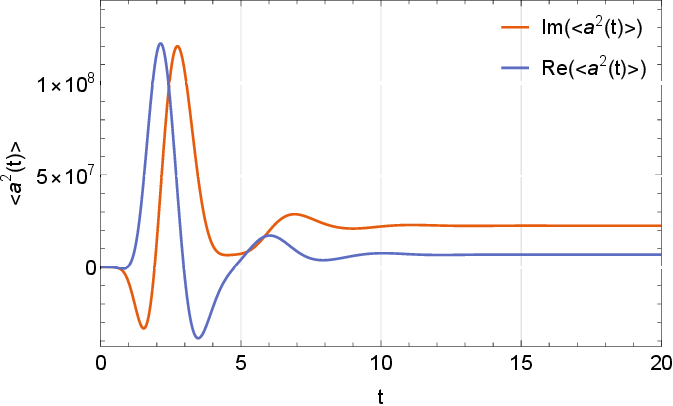}
        \caption{Average of squared creation and annihilation operators}
        \label{fig:stronga2}
    \end{subfigure}
    \caption{Dynamics of first and second moments of creation and annihilation operators in the strong weak field. It demonstrates equally rapid growth and decay, being well localized in time.}
    \label{fig:strong}
\end{figure*}

\section{Strong field\label{sec:StFd}}
	In this section we show that our method allows us to consider the system in any mode, not only in the weak-field regime.
    
    In order to do that we consider the opposite case to the one discussed in the previous section. In this case, we will zero out all terms in which degree $F^{3 - n} = 0$, for $n > 0$. 

    In contrast to the weak-field regime, in this regime we will not use experimental data, but will immediately consider the parameters satisfying the given conditions.
	
	In this case, the action of the conjugate generator on an expression of type ${a^{\dagger}}^na^m$ has the form
	
	\begin{eqnarray}
		\label{eq:GenStrongfield}
		\mathcal{L}^*(t){a^{\dagger}}^na^m &=& -i\dfrac{\chi}{2}\left(c_1(t){a^{\dagger}}^na^{m + 1} + c_2(t){a^{\dagger}}^{n + 1}a^m \right.\nonumber\\
        &&+\left. c_3(t){a^{\dagger}}^na^{m - 1} + c_4(t){a^{\dagger}}^{n - 1}a^m\right),
	\end{eqnarray}
	where
	\begin{eqnarray}
		c_1(t) &=& 2(\beta_2(t))^2\beta_1(t)\alpha_1(t) - 2(\beta_3(t))^2\beta_4(t) \alpha_{42}(t) \nonumber \\
        & &-2(\beta_3(t))^2\beta_4(t) \alpha_{41}(t), \\
		c_2(t) &=&2(\beta_1(t))^2 \beta_2(t) \alpha_{22}(t) - 2 (\beta_4(t))^2 \beta_3(t) \alpha_3(t) \nonumber \\
        & &+ 2\beta_1(t)^2 \beta_2(t) \alpha_{21}(t), \\
		c_3(t) &=& 2m(\beta_1(t))^2 \beta_2(t) \alpha_{21}(t), \\
		c_4(t) &=& -2n (\beta_3(t)^2 \beta_4(t) \alpha_{41}(t).
	\end{eqnarray}
    
	The plots of the mean particle number, linear operators and quadratic operators for the parameters $\ga = 3, \la = 10^{-2}, F = 1.4 \cdot10^2, \DE = 1.5, \chi = 1$ are shown in Fig. \ref{fig:strongn}, \ref{fig:stronga}, \ref{fig:stronga2}.

	It is immediately apparent that in this regime the number of particles reaches a very large value of the order of $10^9$. It is also worth noting the fact that although the dynamics of damping of the number of particles is visible, it has oscillations, the number of particles has maxima and minima in time.

    \section{Conclusions}
	In this work, we have focused on obtaining the Gaussian approximation for the one-mode Kerr model in the external classical field with dissipation. But let us remark that the proposed technique is general and can be applied to multi-mode dissipative models, whose free dynamics preserves the Gaussian states. The specific choice of the model allows us to verify it in the regime, when it is exactly solvable with non-Gaussian solutions. So it can be used as a benchmark of our approach in such an ''uncomfortable'' case for it. And, at the same time, our approach allows us to obtain the dynamics for this model in the case, when it is not exactly solvable. This demonstrates the usefulness of our technique.

    We have considered the weak and strong field regimes, but this is done solely to present the formulas for the conjugate generator in an explicit and comprehensible form. For an arbitrary external field, all computations  still can be performed by computer algebra systems, such as Wolfram Mathematica, but the expressions are too long to include them in the text of the work.However, the very possibility of obtaining such analytical expressions allows asymptotic approximations to be obtained in other parameter regions as well.     
	
	As an obvious direction for the future study we should mention the application of our approach to other specific models, especially  multimode ones. Most of our general calculations are already done in multimode case. And the specific calculations for the discussed model do not seem to complicate much in the multimode case for any reasonable finite number of modes. 
	
	The other direction for the future study is to approximate the solution by the finite linear combinations of Gaussian states and applying our equations to each term in the combination separately. In case of non-convex linear combinations, it may allow one to take some non-Gaussian quantum effects into account.

	\section{Acknowledgments}
	
	The authors are grateful to A.Yu. Karasev, A.M. Savchenko, R. Singh, S.D. Mostovoy, and A.S.~Trushechkin for discussion of the problems discussed in the paper.
	
	The work of K.Sh. Meretukov was supported by the Foundation for the Advancement of Theoretical Physics and Mathematics “BASIS” under the grant № 23-2-1-34-1. 

    \bibliographystyle{unsrt}
    \bibliography{ref}

\begin{thebibliography}{10}

\bibitem{KerrIntr}
Dusan Lorenc and Zhanybek Alpichshev.
\newblock Dispersive effects in ultrafast nonlinear phenomena: The case of optical kerr effect.
\newblock {\em Physical Review Research}, 6(1):013042, 2024.

\bibitem{Beaulieu2025DPT}
S.~Beaulieu, Y.~Zhang, S.~Rosenblum, C.~Reimer, A.~Blais, and J.~M. Fink.
\newblock Observation of first- and second-order dissipative phase transitions in a two-photon driven kerr resonator.
\newblock {\em Nature Communications}, 16:1954, 2025.

\bibitem{Miao2024Nonreciprocity}
Houquan Miao and G.~S. Agarwal.
\newblock Kerr nonlinearity induced nonreciprocity in dissipatively coupled resonators.
\newblock {\em Physical Review Research}, 6(3):033020, 2024.

\bibitem{Tang2024Multiphoton}
X.~Tang and J.~Huang.
\newblock Tunable multiphoton bundles emission in a kerr-type two-photon jaynes--cummings model.
\newblock {\em Physical Review Research}, 6(3):033247, 2024.

\bibitem{Ravets2025}
S.~Ravets, N.~Pernet, N.~Mostaan, N.~Goldman, and J.~Bloch.
\newblock Thouless pumping in a driven-dissipative kerr resonator array.
\newblock {\em Physical Review Letters}, 134(9):093801, 2025.

\bibitem{Najafabadi2023Quasiclassical}
Mojdeh~S Najafabadi, Andrei~B Klimov, Luis~L S{\'a}nchez-Soto, and Gerd Leuchs.
\newblock Quasiclassical approach to the nonlinear kerr dynamics.
\newblock {\em Optics Communications}, 545:129717, 2023.

\bibitem{Zhang2024KerrCombs}
Xucheng Zhang, Chunxue Wang, Zhibo Cheng, Congyu Hu, Xingchen Ji, and Yikai Su.
\newblock Advances in resonator-based kerr frequency combs with high conversion efficiencies.
\newblock {\em npj Nanophotonics}, 1(1):26, 2024.

\bibitem{Singh25}
Ranjit Singh, Alexander~E. Teretenkov, and Anatoly~V. Masalov.
\newblock Sub-poissonian light in a waveguide kerr medium.
\newblock {\em J. Opt. Soc. Am. B}, 42(11):B10--B16, 2025.

\bibitem{asjad2023joint}
Muhammad Asjad, Berihu Teklu, and Matteo~GA Paris.
\newblock Joint quantum estimation of loss and nonlinearity in driven-dissipative kerr resonators.
\newblock {\em Physical Review Research}, 5(1):013185, 2023.

\bibitem{Andrianov2023KerrSqueezing}
Nikolay Kalinin, Thomas Dirmeier, Arseny~A. Sorokin, Elena~A. Anashkina, Luis~L. Sánchez-Soto, Joel~F. Corney, Gerd Leuchs, and Alexey~V. Andrianov.
\newblock Quantum-enhanced interferometer using kerr squeezing.
\newblock {\em Nanophotonics}, 12(14):2945--2952, 2023.

\bibitem{Andrianov:24}
Alexey~V. Andrianov, Alexey~N. Romanov, Arseny~A. Sorokin, Elena~A. Anashkina, Nikolay Kalinin, Thomas Dirmeier, Luis~L. S\'{a}nchez-Soto, and Gerd Leuchs.
\newblock Polarization squeezing in chalcogenide fibers.
\newblock {\em Opt. Lett.}, 49(23):6661--6664, Dec 2024.

\bibitem{Anashkina:20}
Elena~A. Anashkina, Alexey~V. Andrianov, Joel~F. Corney, and Gerd Leuchs.
\newblock Chalcogenide fibers for kerr squeezing.
\newblock {\em Opt. Lett.}, 45(19):5299--5302, Oct 2020.

\bibitem{MaslovN}
Junqiu Liu, Guanhao Huang, Rui~Ning Wang, Jijun He, Arslan~S Raja, Tianyi Liu, Nils~J Engelsen, and Tobias~J Kippenberg.
\newblock High-yield, wafer-scale fabrication of ultralow-loss, dispersion-engineered silicon nitride photonic circuits.
\newblock {\em Nature communications}, 12(1):2236, 2021.

\bibitem{MaslovSolve}
SP~Nikitin and AV~Masalov.
\newblock Quantum state evolution of the fundamental mode in the process of second-harmonic generation.
\newblock {\em Quantum Optics: Journal of the European Optical Society Part B}, 3(2):105, 1991.

\bibitem{corney2015non}
JF~Corney and MK~Olsen.
\newblock Non-gaussian pure states and positive wigner functions.
\newblock {\em Physical Review A}, 91(2):023824, 2015.

\bibitem{joneckis1993quantum}
Lance~G Joneckis and Jeffrey~H Shapiro.
\newblock Quantum propagation in a kerr medium: lossless, dispersionless fiber.
\newblock {\em JOSA B}, 10(6):1102--1120, 1993.

\bibitem{banerjee1993interaction}
Arup Banerjee.
\newblock Interaction of squeezed light with a kerr medium.
\newblock {\em Quantum Optics: Journal of the European Optical Society Part B}, 5(1):15, 1993.

\bibitem{genoni2009enhancement}
Marco~G Genoni, Carmen Invernizzi, and Matteo~GA Paris.
\newblock Enhancement of parameter estimation by kerr interaction.
\newblock {\em Physical Review A—Atomic, Molecular, and Optical Physics}, 80(3):033842, 2009.

\bibitem{roman2015parametric}
R~Rom{\'a}n-Ancheyta, M~Berrondo, and J~R{\'e}camier.
\newblock Parametric oscillator in a kerr medium: evolution of coherent states.
\newblock {\em JOSA B}, 32(8):1651--1655, 2015.

\bibitem{bolandhemmat2023quantum}
E~Bolandhemmat and F~Kheirandish.
\newblock Quantum dynamics of a driven parametric oscillator in a kerr medium.
\newblock {\em Scientific Reports}, 13(1):9056, 2023.

\bibitem{MankoManko2003PhotonNumberTomography}
O.~V. Man’ko and V.~I. Man’ko.
\newblock Photon‑number tomography of multimode states and positivity of the density matrix.
\newblock {\em Journal of Russian Laser Research}, 24:497--506, 2003.

\bibitem{MankoManko2009PhotonOptical}
O.~V. Man’ko and V.~I. Man’ko.
\newblock Photon number and optical tomograms for gaussian states.
\newblock {\em Laser Physics}, 19:1804--1808, 2009.

\bibitem{MankoSharapovShchukin2001}
I.~Man'ko V.\, A.~Sharapov V.\, and V.~Shchukin E.\.
\newblock Tomography of multimode quantum systems with quadratic hamiltonians and multivariable hermite polynomials.
\newblock {\em Journal of Russian Laser Research}, 22:410--436, 2001.

\bibitem{mele2024learningquantumstatescontinuous}
Francesco~Anna Mele, Antonio~Anna Mele, Lennart Bittel, Jens Eisert, Vittorio Giovannetti, Ludovico Lami, Lorenzo Leone, and Salvatore F.~E. Oliviero.
\newblock Learning quantum states of continuous variable systems, 2024.

\bibitem{graefe2018lindblad}
EM~Graefe, B~Longstaff, T~Plastow, and R~Schubert.
\newblock Lindblad dynamics of gaussian states and their superpositions in the semiclassical limit.
\newblock {\em Journal of Physics A: Mathematical and Theoretical}, 51(36):365203, 2018.

\bibitem{LopezSaldivar2021}
Julio~A. López‑Saldívar, Margarita~A. Man’ko, and Vladimir~I. Man’ko.
\newblock Measurement of the temperature using the tomographic representation of thermal states for quadratic hamiltonians.
\newblock {\em Entropy}, 23(11):1445, 2021.

\bibitem{GaussState}
Rapha{\"e}l Menu and Tommaso Roscilde.
\newblock Gaussian-state ansatz for the non-equilibrium dynamics of quantum spin lattices.
\newblock {\em SciPost Physics}, 14(6):151, 2023.

\bibitem{Kholevo}
Alexander~S Holevo.
\newblock {\em Quantum systems, channels, information: a mathematical introduction}.
\newblock Walter de Gruyter GmbH \& Co KG, 2019.

\bibitem{TrushechkinVolovich2009}
A.~S. Trushechkin and I.~V. Volovich.
\newblock Functional classical mechanics and rational numbers.
\newblock {\em {P-Adic Numbers, Ultrametric Analysis and Applications}}, 1(4):361--367, 2009.

\bibitem{Volovich2011}
Igor~V. Volovich.
\newblock Randomness in classical mechanics and quantum mechanics.
\newblock {\em Foundations of Physics}, 41(3):516--528, 2011.

\bibitem{kawasaki1973theory}
Kyozi Kawasaki and James~D Gunton.
\newblock Theory of nonlinear transport processes: Nonlinear shear viscosity and normal stress effects.
\newblock {\em Physical Review A}, 8(4):2048, 1973.

\bibitem{rau1996reversible}
Jochen Rau and Berndt M{\"u}ller.
\newblock From reversible quantum microdynamics to irreversible quantum transport.
\newblock {\em Physics Reports}, 272(1):1--59, 1996.

\bibitem{semin2020dynamical}
Vitalii Semin and Francesco Petruccione.
\newblock Dynamical and thermodynamical approaches to open quantum systems.
\newblock {\em Scientific reports}, 10(1):2607, 2020.

\bibitem{meretukov2024time}
Kh~Sh Meretukov and AE~Teretenkov.
\newblock On time-dependent projectors and a generalization of the thermodynamical approach in the theory of open quantum systems.
\newblock {\em Proceedings of the Steklov Institute of Mathematics}, 324(1):135--152, 2024.

\bibitem{Review}
AE~Teretenkov.
\newblock Irreversible quantum evolution with quadratic generator.
\newblock {\em Infinite Dimensional Analysis, Quantum Probability and Related Topics}, 22(04):1930001, 2019.

\bibitem{Teretenkov2017drift}
A.~E. Teretenkov.
\newblock Quadratic dissipative evolution of gaussian states with drift.
\newblock {\em Mathematical Notes}, 101(1--2):341--351, 2017.

\bibitem{BreuerPetruccione2002}
Heinz-Peter Breuer and Francesco Petruccione.
\newblock {\em The Theory of Open Quantum Systems}.
\newblock Oxford University Press, Oxford, 2002.

\bibitem{KarasevTeretenkov2023}
A.~Yu. Karasev and A.~E. Teretenkov.
\newblock Time-convolutionless master equations for composite open quantum systems.
\newblock {\em Lobachevskii Journal of Mathematics}, 44(6):2051--2064, 2023.

\bibitem{bogoliubov1946problems}
N.~N. Bogoliubov.
\newblock {\em {Problems of Dynamical Theory in Statistical Physics}}.
\newblock {Gostekhisdat}, {Moscow}, 1946.

\bibitem{van1955energy}
Leon Van~Hove.
\newblock Energy corrections and persistent perturbation effects in continuous spectra.
\newblock {\em Physica}, 21(6-10):901--923, 1955.

\bibitem{teretenkov2021non}
Alexander~E Teretenkov.
\newblock Non-perturbative effects in corrections to quantum master equations arising in bogolubov--van hove limit.
\newblock {\em Journal of Physics A: Mathematical and Theoretical}, 54(26):265302, 2021.

\bibitem{bogaevski2012algebraic}
Vladamir~Nikolaevich Bogaevski and Aleksandr Povzner.
\newblock {\em Algebraic methods in nonlinear perturbation theory}, volume~88.
\newblock Springer Science \& Business Media, 2012.

\bibitem{basharov2020global}
AM~Basharov.
\newblock “global” and “local” approaches to the theory of open quantum optical systems.
\newblock {\em Journal of Experimental and Theoretical Physics}, 131:853--875, 2020.

\bibitem{trubilko2020hierarchy}
Andrei~Igorevich Trubilko and Askhat~Maskhudovich Basharov.
\newblock Hierarchy of times of open optical quantum systems and the role of the effective hamiltonian in the white noise approximation.
\newblock {\em JETP Letters}, 111:532--538, 2020.

\bibitem{basharov2021effective}
AM~Basharov.
\newblock The effective hamiltonian as a necessary basis of the open quantum optical system theory.
\newblock In {\em Journal of Physics: Conference Series}, volume 1890, page 012001. IOP Publishing, 2021.

\bibitem{trubilko2021effective}
AI~Trubilko and AM~Basharov.
\newblock Effective quantum oscillator of a cavity with oscillating parameters.
\newblock {\em Journal of Experimental and Theoretical Physics}, 132:216--222, 2021.

\bibitem{basharov2024atom}
AM~Basharov.
\newblock Atom--photon cluster in nonlinear and quantum optics.
\newblock {\em Bulletin of the Russian Academy of Sciences: Physics}, 88(6):835--841, 2024.

\bibitem{DodonovManko1988}
V.~V. Dodonov and V.~I. Man'ko.
\newblock Evolution equations for the density matrices of linear open systems.
\newblock In A.~A. Komar, editor, {\em Classical and Quantum Effects in Electrodynamics}, volume 176 of {\em Proceedings of the Lebedev Physics Institute}, pages 53--60. Nova Science, Commack, New York, 1988.

\bibitem{DodonovManko1985}
V.~V. Dodonov and O.~V. Man'ko.
\newblock Quantum damped oscillator in a magnetic field.
\newblock {\em Physica A}, 130(1--2):353--366, 1985.

\bibitem{HeinosaariHolevoWolf2010}
Teiko Heinosaari, Alexander~S. Holevo, and Michael~M. Wolf.
\newblock The semigroup structure of gaussian channels.
\newblock {\em Quantum Information \& Computation}, 10(7--8):619--635, 2010.

\bibitem{LopezSaldivar2020}
Julio~A. López‑Saldívar, Margarita~A. Man’ko, and Vladimir~I. Man’ko.
\newblock Differential parametric formalism for the evolution of gaussian states: Nonunitary evolution and invariant states.
\newblock {\em Entropy}, 22(5), 2020.

\bibitem{Teretenkov2016}
A.~E. Teretenkov.
\newblock Quadratic dissipative evolution of gaussian states.
\newblock {\em Mathematical Notes}, 100(4):642--646, 2016.

\bibitem{LopezSaldivarManoko2021}
Julio~A L{\'o}pez-Sald{\'\i}var, Margarita~A Man’ko, and Vladimir~I Man’ko.
\newblock Nonlinear differential dynamics of gaussian states.
\newblock In {\em AIP Conference Proceedings}, volume 2362, page 040007. AIP Publishing LLC, 2021.

\bibitem{NosalTeretenkov}
Iu~A Nosal and AE~Teretenkov.
\newblock Higher order moments dynamics for some multimode quantum master equations.
\newblock {\em Lobachevskii Journal of Mathematics}, 43(7):1726--1739, 2022.

\bibitem{DodonovManko2003}
V.~V. Dodonov and V.~I. Man'ko.
\newblock {\em Theory of Nonclassical States of Light}.
\newblock Taylor \& Francis, London and New York, 2003.

\bibitem{ProsenSeligman2010}
T.~Prosen and T.~H. Seligman.
\newblock Quantization over boson operator spaces.
\newblock {\em Journal of Physics A: Mathematical and Theoretical}, 43(39):392004, 2010.

\bibitem{Dis}
Alexander~E. Teretenkov.
\newblock {\em Exactly solvable problems of irreversible quantum evolution}.
\newblock Ph.d. thesis, Lomonosov Moscow State University, Moscow, Russia, 2018.
\newblock In Russian.

\bibitem{KarasevTeretenkov2025}
Artem Karasev and Alexander~E. Teretenkov.
\newblock Parametric approximation as an open-quantum-system problem.
\newblock {\em Physical Review A}, 111(5):052203, 2025.

\bibitem{gaidash2025lindblad}
Andrei Gaidash, Alexei~D Kiselev, Anton Kozubov, and George Miroshnichenko.
\newblock Lindblad dynamics of open multimode bosonic systems: Algebra of quadratic superoperators, exceptional points, and speed of evolution.
\newblock {\em Physical Review A}, 111(6):062211, 2025.

\bibitem{Baranger1958}
Michel Baranger.
\newblock Problem of overlapping lines in the theory of pressure broadening.
\newblock {\em Physical Review}, 111:494, 1958.

\bibitem{Havel2003}
Timothy~F. Havel.
\newblock Robust procedures for converting among lindblad, kraus and matrix representations of quantum dynamical semigroups.
\newblock {\em Journal of Mathematical Physics}, 44(2):534--557, 2003.

\bibitem{sorokin2022towards}
Arseny~A Sorokin, Gerd Leuchs, Joel~F Corney, Nikolay~A Kalinin, Elena~A Anashkina, and Alexey~V Andrianov.
\newblock Towards quantum noise squeezing for 2-micron light with tellurite and chalcogenide fibers with large kerr nonlinearity.
\newblock {\em Mathematics}, 10(19):3477, 2022.

\end{thebibliography}
\appendix

\section{Gaussian states}

    The arbitrary sums of quadratic and linear forms in such operators can be written as $\dfrac{1}{2}\mf{a}^TK\mf{a} + g^T\mf{a}$, where $K \in \mathbb{C}^{2d \times 2d}, g \in \mathbb{C}^{2d}$. The canonical commutation relations are written in the form 

    \[[f^T\mf{a},\mf{a}^Tg] = -f^TJg,\]
    where $f \in \mathbb{C}^{2d}, g \in \mathbb{C}^{2d}$ and the matrix $J$ has the form

    \[J = \begin{pmatrix}
			0 & -I_d \\
			I_d & 0
		\end{pmatrix},\]
    where $I_d$ is an identity matrix.

    The defintion of a Gussian states in terms of characteristic functions \cite[Subsec.~12.3.1]{Kholevo} takes the following form in our notation:
    \begin{definition} 
        $\rho$ is a Gaussian state if
        \begin{equation}
            \operatorname{tr} (\rho \; e^{i \mathbf{z}^T \mathfrak{a}}) = e^{- \frac12 \mathbf{z}^T C \mathbf{z} + i   \mathbf{z}^T m},
        \end{equation}
        where $\mathbf{z} = (z_1, z_2, \ldots, z_d, \overline{z_1}, \overline{z}_2,\ldots, \overline{z}_d)^T \in \mathbb{C}^{2d}$.
    \end{definition}
    
    A genric  Gaussian state can be also represented in the following form \cite{Review}:
    \begin{equation}
		\label{eq:ProjForAv}
		\rho_{ans}(m, C)=\exp \left(\dfrac{1}{2}\mathfrak{a}^TK\mathfrak{a} + g^T\mathfrak{a} + s\right),
	\end{equation}
	where $e^s = \sqrt{|\det(e^{KJ} - I)|}e^{\frac{1}{2}g^TK^{-1}g}$, $ m $ is the vector of mean linear operators and $ C $ is a covariance matrix, defined in terms $ K $ and $ g $ by the formulae in 
	\begin{equation}
		\label{eq:ConOfCFromK}
		C = -\frac{J}{2}\text{coth}\left(\frac{KJ}{2}\right), \qquad m = - K^{-1} g.
	\end{equation}

    In order to obtain an explicit form of this derivative, we need to find the relation between the matrix $K$ and the covariance matrix $C$. From Eq.~\eqref{eq:ConOfCFromK}, expressing $K$ from this equation, we obtain
	\begin{equation}
		\label{eq:ConOfKFromC}
		KJ = 2\;\text{arcoth}\left(-2J^{-1}C\right).		
	\end{equation}

\section{Second order equations}
\label{sec:SecOrdEq}

    In \cite{meretukov2024time} it was shown that for general ansatz $\rho_{ans} (\vec{E} )$, the following second order equations are satisfied
    \begin{eqnarray}
		&\frac{d}{dt} \vec{E}(t) = \lambda   \Tr ( \vec{P}\mathcal{L}(t)\rho_{ans} (\vec{E}(t) ) \nonumber\\
        &+ \lambda^2 \Tr \left( \vec{P} \mathcal{L}(t)  \int_{t_0}^t dt_1 \mathcal{L}(t_1)\rho_{ans} (\vec{E}(t))   \right) \nonumber \\
		& 	-  \lambda^2 \Tr\Biggl( \vec{P} \left(\Tr (  \vec{P} \mathcal{L}(t)\rho_{ans} (\vec{E}) ) , \frac{\partial \rho_{ans}(\vec{E})}{\partial \vec{E}} \right) \times \label{eq:secOrderEqE}\\
		& \qquad \times \left(\Tr (  \vec{P}\int_{t_0}^t dt_1\mathcal{L}(t_1)\rho_{ans} (\vec{E}) )  , \frac{\partial \rho_{ans}(\vec{E})}{\partial \vec{E}} \right) \Biggr)_{\vec{E} =\vec{E}(t) }, \nonumber
	\end{eqnarray}
    where $\vec{E}(t) = \Tr\rho(t)\vec{P}$ are the averages of the relevant operators. They parameterize the ansatz $ \rho_{ans} (\vec{E}(t) ) $. In our case, the relevant operators $ \vec{P} $ are formed by elements of $\mf{a} $ and $ \mf{a} \mf{a}^T $, but in this work we use parametrization $ (m,C) $ instead of $ \vec{E} $.	As can be seen, the first order coincides with the equation in the Heisenberg representation, assuming that the ansatz holds.

    Thus, $\vec{E} = (m, S)$, where according to Eqs.~\eqref{eq:mDef}--\eqref{eq:CDef}
    \begin{align*}
             m =& \Tr \mathfrak{a}  \rho_{ans}(m, C), \\
             S:=& \Tr \mathfrak{a} \mathfrak{a}^T \rho_{ans}(m, C) = m m^T + C - \frac{J}{2}.
    \end{align*}
    Then we have
    \begin{align*}
         &\frac{\partial}{\partial m} \rho_{ans}\left(m, S - m m^T + \frac{J}{2} \right) \\
         &=  \frac{\partial}{\partial m} \rho_{ans}(m, C ) - 2 \left(\frac{\partial}{\partial C} \rho_{ans}(m, C) \right)m,\\
         &\frac{\partial}{\partial S} \rho_{ans}\left(m, S - m m^T + \frac{J}{2} \right) 
         =  \frac{\partial}{\partial C} \rho_{ans}(m, C ) .
    \end{align*}

\begin{widetext}
     Then we have
    \begin{align*}
        R_t(m,C) := \left(\Tr (  \vec{P} \mathcal{L}(t)\rho_{ans} (\vec{E}) ) , \frac{\partial \rho_{ans}(\vec{E})}{\partial \vec{E}} \right)
        = \Tr  \mathfrak{a}^T (\mathcal{L}(t)\rho_{ans} (m, C) ) 
        \left( \frac{\partial}{\partial m} \rho_{ans}(m, C ) - 2 \left(\frac{\partial}{\partial C} \rho_{ans}(m, C) \right)m \right)\\
        +  \Tr   \mathfrak{a_i} \mathfrak{a}_j  (\mathcal{L}(t)\rho_{ans} (m, C) )  \frac{\partial}{\partial C_{ij}} \rho_{ans}(m, C) 
    \end{align*}
    and Eq.~\eqref{eq:secOrderEqE}  takes the form (we omit dependence on $t$ for simplicity)
    \begin{equation}
		\frac{d}{dt} m = \lambda   \Tr ( \mathfrak{a} \mathcal{L}(t) \rho_{ans}(m, C)) + \lambda^2 \left(\Tr \left( \mathfrak{a} \mathcal{L}(t)  \int_{t_0}^t dt_1 \mathcal{L}(t_1) \rho_{ans}(m, C)   \right) 	-   \Tr\left( \mathfrak{a} R_t(m,C) \int_{t_0}^t dt_1 R_{t_1}(m,C) \right) \right)
	\end{equation}
    \begin{equation}
		\frac{d}{dt} S = \lambda   \Tr ( \mathfrak{a} \mathfrak{a}^T\mathcal{L}(t) \rho_{ans}(m, C)) + \lambda^2 \left(\Tr \left( \mathfrak{a} \mathfrak{a}^T \mathcal{L}(t)  \int_{t_0}^t dt_1 \mathcal{L}(t_1) \rho_{ans}(m, C)   \right) 	-   \Tr\left( \mathfrak{a} \mathfrak{a}^T R_t(m,C) \int_{t_0}^t dt_1 R_{t_1}(m,C) \right) \right)
	\end{equation}
    Taking into account 
    \begin{equation*}
        \frac{d}{dt} C = \frac{d}{dt} S - \left( \frac{d}{dt} m \right) m^T - m \left( \frac{d}{dt} m^T \right)
    \end{equation*}
    and $\mathfrak{A}(m) := \mathfrak{a} \mathfrak{a}^T -m \mathfrak{a}^T - \mathfrak{a} m^T$ we have
    \begin{equation*}
        \frac{d}{dt} C = \lambda   \Tr \mathfrak{A}(m) \mathcal{L}(t) \rho_{ans}(m, C)) + \lambda^2 \left(\Tr \left( \mathfrak{A}(m) \mathcal{L}(t)  \int_{t_0}^t dt_1 \mathcal{L}(t_1) \rho_{ans}(m, C)   \right) 	-   \Tr\left( \mathfrak{A}(m) R_t(m,C) \int_{t_0}^t dt_1 R_{t_1}(m,C) \right) \right) 
    \end{equation*}
    Using the definition of adjoint generator $ \Tr X\mathcal{L}(t) Y = \Tr( \mathcal{L}^{*}(t)X )Y  $ (if $\mathcal{L}(t)$ is Hermicity preserving) we obtain
\end{widetext}

	\section{Derivative of the Gaussian ansatz\label{App}}

    \begin{theorem}
		\label{th:DerOfGaus}
		The derivative of a Gaussian state can be expressed in the following form 
		\begin{equation}\label{eq:DerOfGaus}
			\frac{\partial}{\partial x}\rho_{ans}(m,C) = \left(\frac{1}{2}\mf{a}^TM\mf{a} + \mf{a}^TG + c\right)\rho_{ans}(m,C),
		\end{equation}
		where
		\begin{eqnarray}
            M &=& \frac{I}{C - \frac{J}{2}}\left(\frac{\partial}{\partial x}C\right)\frac{I}{C + \frac{J}{2}},\\
            G &=& \left(  I + 2J_-  \right)\frac{\partial}{\partial x}\left(\frac{I}{C + \frac{J}{2}}m\right),\\
            c &=&-m^TJ_-\frac{\partial}{\partial x}\left(\frac{I}{C + \frac{J}{2}}m \right) \nonumber\\
            &-& \frac{1}{2}\frac{\partial}{\partial x}\left(m^T\left( J_- + J_+ + KJ \right)J^{-1}m + s\right),
        \end{eqnarray}
		where $m = (\braket{a} , \braket{a^{\dagger}})^T$, $C$ --- covariance matrix and $J_{\pm} = \dfrac{I}{2J^{-1}C \pm I}$.
	\end{theorem}
    \renewcommand\qedsymbol{$\blacksquare$}
	\begin{proof}
		The formula \eqref{eq:DerOfGaus} for the derivative of a Gaussian state in  form \eqref{eq:ProjForAv} can be found in \cite{Review} with 
		\begin{eqnarray*}
			M &= &\frac{1}{2}e^{-KJ}\frac{\partial}{\partial x}e^{KJ}J^{-1},\\
			G &= &e^{-KJ}\frac{\partial}{\partial x}\left(\frac{e^{KJ} - I}{KJ}g\right),\\
			c &=& \frac{1}{2}g^TJ\frac{e^{-KJ} - I}{KJ}\frac{\partial}{\partial x}\left( \frac{e^{KJ} - I}{KJ}g  \right) \\
            &+&
        \frac{\partial}{\partial x}\left(  \frac{1}{2}g^TJ\frac{I}{KJ}(\text{sh}(KJ) - KJ)\frac{1}{KJ}g  +s \right).
		\end{eqnarray*}
		
		Using the formula $\text{arcoth}(x) = \frac{1}{2}\ln\left(\frac{x + 1}{x - 1}\right)$, we can explicitly write $e^{KJ}$ as
		\begin{equation*}
			e^{KJ} = \frac{2J^{-1}C - I}{2J^{-1}C + I} = I - \frac{2I}{2J^{-1}C + I}.
		\end{equation*}

            The same can be done for $e^{-KJ}$
            \begin{equation*}
			e^{-KJ} = \frac{2J^{-1}C + I}{2J^{-1}C - I} = I + \frac{2I}{2J^{-1}C - I},
		\end{equation*}
            thus
            \begin{equation}
                \label{eq:exppmKJ}
                e^{\pm KJ} = I \mp 2J_{\pm}
            \end{equation}

            Now we can calculate the expression $\frac{e^{KJ}-I}{KJ}g$

            \begin{eqnarray}
            \label{eq:fractionofexpKJ}
                \frac{e^{KJ} - I}{KJ}g &=&  \frac{-2J_+}{KJ} g =
                 -2J_+J^{-1}K^{-1}(-Km)   =\\&=&  2J_+J^{-1}m =
               \frac{2I}{2J^{-1}C + I} J^{-1}m = \frac{I}{C + \frac{J}{2}}m.\nonumber
            \end{eqnarray}

            The last expression will be $g^TJ\frac{e^{-KJ} - I}{KJ}$. Similarly with Eq.~\eqref{eq:fractionofexpKJ} we obtain the following result

            \begin{eqnarray}
                \label{eq:fractionofexp-KJ}
                g^TJ\frac{e^{-KJ} - I}{KJ} = -2m^TJ_-.
            \end{eqnarray}

            Using the results of Eq.~\eqref{eq:exppmKJ} and the formula $\frac{\partial}{\partial x}X^{-1} = -X^{-1}\left(\frac{\partial}{\partial x}X\right)X^{-1}$ in the expression $e^{-KJ}\frac{\partial}{\partial x}e^{KJ}J^{-1}$ we obtain an expression for the matrix $M$

            \begin{equation*}
                M = \frac{1}{2}e^{-KJ}\frac{\partial}{\partial x}e^{KJ}J^{-1} = \frac{1}{2}\frac{I}{C - \frac{J}{2}}\left(\frac{\partial}{\partial x}C\right)\frac{I}{C + \frac{J}{2}}.
            \end{equation*}

            Using the results of Eq.~\eqref{eq:exppmKJ} and Eq.~\eqref{eq:fractionofexpKJ} we obtain an expression for $G$

            \begin{equation*}
                G = (I+2J_-)\frac{\partial}{\partial x}\left( \frac{I}{C+\frac{J}{2}}m  \right).
            \end{equation*}
	
		Using the results of Eq.~\eqref{eq:exppmKJ} and Eq.~\eqref{eq:fractionofexpKJ} and Eq.~\eqref{eq:fractionofexp-KJ} we obtain an expression for $c$

            \begin{eqnarray*}
                c &=&-m^TJ_-\frac{\partial}{\partial x}\left(\frac{I}{C + \frac{J}{2}}m \right.\\
            &-& \left.\frac{1}{2}m^T\left( J_- + J_+ + KJ \right)J^{-1}m + s\right).
            \end{eqnarray*}
		
            It remains to calculate the derivative of the normalization. The normalization formula has the following form
        
		\begin{equation*}
			\label{eq:c}
			e^{s} = \left(\Tr e^{\frac{1}{2} \mathfrak{a}^TK\mathfrak{a} + g^T\mf{a}}\right)^{-1} = \sqrt{\left|\text{det}(e^{KJ} - I)\right|}e^{\frac{1}{2}g^TK^{-1}g}.
		\end{equation*}
		
		However, as we see, in the derivative s stands the modulus from the determinant. Let us consider this determinant in more detail
		
		\begin{eqnarray*}
			e^{KJ} - I &=& -\dfrac{2I}{2J^{-1}C + I} = - \dfrac{2I}{J^{-1}(2C + J)} =\\&=& -\dfrac{I}{J^{-1}(C+\frac{J}{2})} = -\dfrac{I}{J^{-1}D^T},
		\end{eqnarray*}
		
		since all matrices of size $2d \times 2d$, the multiplier $-1$ makes no contribution. Thus
		
		\begin{equation*}
			\det(e^{KJ} - I) = \dfrac{\det(J)}{\det(D)}
		\end{equation*}
		we know, that $DE \geq 0$, hence $\det(D)\det(E) \geq 0$, then the sign of $\det(D)$ coincides with the sign of $\det(E)$. In our case $d = 2$, $\det(J) = - \det(E) = 1$ so $\det(e^{KJ} - I) \leq 0$. We obtain that when the modulus is removed, the sign must be $-1$.
	\end{proof}

    Taking into account the definitions of $J_{\pm}$, $s$ and Eq.~\eqref{eq:ConOfKFromC} we obtain Eqs.~\eqref{eq:DerWithLin}--\eqref{eq:sInC}.

    \section{Passage of the Gaussian ansatz through a quadratic form\label{app:anzoverform}}

    \begin{theorem}
        \label{th:anzoverform}
        If we carry the Gaussian ansatz through the quadratic form, we obtain the Eq.~\eqref{eq:DenMatrOverForm},
        where
        \begin{align*}
    		M' &= \frac{1}{2}\left(I + 2J_-\right)MJ\left(I - 2J_+\right)J^{-1}, \\
    		f' &= 2\left(I + 2J_-\right)MJJ_+J^{-1}m + \left(I + 2J_-\right)f, \\
    		c' &= 2\left(f^T - m^TJ_-M\right)JJ_+J^{-1}m + c.
    	\end{align*}
        \begin{proof}
        Using the results of \cite{Review} we obtain \eqref{eq:DenMatrOverForm} with
    	\begin{align*}
    		M' &= \frac{1}{2}e^{-KJ}Me^{JK}, \\
    		f' &= e^{-KJ}M\frac{e^{JK} - I}{JK}Jg + e^{-KJ}f, \\
    		c' &= \left( \frac{1}{2}g^TJ\frac{e^{-KJ} - I}{KJ}M +f^T  \right)\frac{e^{JK} - I}{JK}Jg + c.
    	\end{align*}

        The main expressions were calculated in Appendix~\ref{App}, in order to simplify these expressions we need to calculate $e^{JK}$. Using the relation \eqref{eq:ConOfKFromC} we obtain that $JK = 2J\text{arccoth}(-2J^{-1}C)J^{-1}$, hence $e^{JK} = Je^{\text{arccoth}(-2J^{-1}C)}J^{-1}$. Thus

        \begin{eqnarray*}
            e^{JK} &=& Je^{\text{arccoth}(-2J^{-1}C)}J^{-1} = J\left( \frac{2J^{-1}C - I}{2J^{-1}C + I}\right)J^{-1} \\&=& I - 2JJ_+J^{-1},
        \end{eqnarray*}
        then
        \begin{eqnarray}
        \label{eq:ejk}
           \frac{e^{JK} - I}{JK}g = 2JJ_+J^{-1}m. 
        \end{eqnarray}

        Substituting the results of App.~\ref{App} and Eq.~\eqref{eq:ejk} into Eq.~\eqref{eq:DenMatrOverForm} we obtain the following expressions for matrices
        \begin{align*}
    		M' &= \frac{1}{2}\left(I + 2J_-\right)MJ\left(I - 2J_+\right)J^{-1}, \\
    		f' &= 2\left(I + 2J_-\right)MJJ_+J^{-1}m + \left(I + 2J_-\right)f, \\
    		c' &= 2\left(f^T - m^TJ_-M\right)JJ_+J^{-1}m + c.
    	\end{align*}
        \end{proof}
    Taking into account the definition of $J_{\pm}$ we obtain \eqref{eq:DenMatrOverForm}--\eqref{eq:cPrimec}. 

      \textnormal{Now let us find the square of the ansatz. Since we use a Gaussian ansatz, its square remains Gaussian, but only with renormalized covariance matrix $C'$. }

        \begin{theorem}
		\begin{equation*}
			(\rho_{ans}(m, C))^2 = \frac{1}{\sqrt{|\det(2 C)|}} \rho_{ans}(m, C'),
		\end{equation*}
		where 
		\begin{equation*}
			C' = - \frac{J}{4}\left(\frac{1}{2 J C} + 2 J C\right).
		\end{equation*}
    	\end{theorem}
    	\begin{proof}
    		\begin{equation*}
    			\rho_{m,C}^2 = \frac{Z'}{Z^2}\rho_{m',C'},
    		\end{equation*}
    		where 
    		
    		\begin{equation*}
    			Z = \operatorname{tr} e^{\frac12 \mathfrak{a}^T K  \mathfrak{a} + g\mf{a}}  = \frac{1}{\sqrt{|\det(e^{KJ} - I)|}}e^{\frac{1}{2}g^TK^{-1}g},
    		\end{equation*}
    		
    		\begin{equation*}
    			Z' = \operatorname{tr} e^{\frac12 \mathfrak{a}^T 2K  \mathfrak{a} + 2g\mf{a}}  = \frac{1}{\sqrt{|\det(e^{2KJ} - I)|}}e^{g^TK^{-1}g}
    		\end{equation*}
    		and

                \begin{equation*}
                    m' = -K'^{-1}g' = -(2K)^{-1}2g = -K^{-1}g.
                \end{equation*}
                
    		Then 
    		
    		\begin{equation*}
    			\frac{\operatorname{tr} e^{\mathfrak{a}^T K  \mathfrak{a} + 2g\mf{a}}}{(\operatorname{tr} e^{\frac12 \mathfrak{a}^T K  \mathfrak{a} + g\mf{a}})^2} = \sqrt{\left|\det\left(\frac{(e^{KJ} - I)^2}{e^{2 KJ} - I}\right)\right|},
    		\end{equation*}
    		
    		\begin{equation*}
    			\frac{(e^{KJ} - I)^2}{e^{2 KJ} - I} = \frac{e^{KJ} - I}{e^{KJ} + I} = \frac{1}{\coth \frac{KJ}{2}} = \frac{1}{-2J^{-1} C},
    		\end{equation*}
                thus
            
    		\begin{equation*}
    			\frac{\operatorname{tr} e^{\frac12 \mathfrak{a}^T 2K  \mathfrak{a}}  }{(\operatorname{tr} e^{\frac12 \mathfrak{a}^T K  \mathfrak{a}})^2} = \sqrt{\left|\det\left(\frac{1}{-2J^{-1} C}\right)\right|} = \frac{1}{\sqrt{|\det(2 C)|}}.
    		\end{equation*}

                In order to find the relationship of the new covariance matrix $C'$ with the old $C$, we apply the trigonometric formula $\coth(2x) = \frac12\left(\coth(x) + \frac{1}{\coth(x)}\right)$, then
    		\begin{eqnarray*}
    			C' &=& C(K') = C(2K) = -\frac{J}{2}\coth\left(\frac{K'J}{2}\right) =\\
                &=& -\frac{J}{2}\coth\left(2\left(\frac{KJ}{2}\right)\right) = - \frac{J}{4}\left(\frac{1}{2 J C} + 2 J C\right).
    		\end{eqnarray*}
    	\end{proof}
    \end{theorem}

    \section{Generator in the Heisenberg representation\label{Heiz}}
    \begin{theorem}
		\label{th:HeiGenerator}
		The generator $\mathcal{L}(t)\rho = -i[H_{\lambda},\rho] + \dfrac{\ga}{2}(a\rho a^{\dagger} - \dfrac12\{a^{\dagger}a,\rho\}),$ in the Heisenberg representation has the form 
		\begin{equation*}
			\mathcal{L}^*\rho = i[H_{\la},\rho]  + \dfrac{\ga}{2}(a^{\dagger}\rho a - \dfrac{1}{2}\{a^{\dagger}a,\rho\}\}).
		\end{equation*}
	\end{theorem}
	
	\begin{proof}
		The operator $\mathcal{L}$ can be represented as $\mathcal{L} = \sum X_i\cdot Y_i$, then 
		
		\begin{equation*}
			\Tr(A\mathcal{L}B) = \sum\Tr(AX_i(B)Y_i).
		\end{equation*}
		
		Since the trace is invariant with respect to cyclic permutations, then 	
		
		\begin{eqnarray*}
			\sum\Tr(AX_i(B) Y_i) &=& \sum\Tr(B Y_i(A)X_i) =\\&=& \sum\Tr(B(Y_i\cdot X_i)(A)),
		\end{eqnarray*}
		hence the operator $\mathcal{L}^*$ is represented in the form $\mathcal{L}^* = \sum Y_i \cdot X_i$.

		The operator $\mathcal{L}$ can be rewritten as
		
		\begin{eqnarray*}
			\mathcal{L}(B) &=& -i(HB\mathcal{I} - \mathcal{I}B H) \\&
            +& \dfrac{\ga}{2}(aB a^{\dagger}
            - \dfrac{1}{2}[a^{\dagger}aB\mathcal{I} + \mathcal{I}B a^{\dagger}a]),
		\end{eqnarray*}
		hence
		
		\begin{equation*}
			Y(A)X = -iAH + iHA + \dfrac{\ga}{2}\left(a^{\dagger}Aa - \dfrac{1}{2}[Aa^{\dagger}a + a^{\dagger}aA]\right).
		\end{equation*}
        
		Thus, we get the desired result.
	\end{proof}

    \section{Connection of experimental coefficients with theoretical ones. \label{Parrams}}

    In order to find the relationship between experimental data and theoretical data, let us consider equation (7) from the \cite{sorokin2022towards}. In this equation, let us neglect the memory effects, i.e. in the formula for $R(t)$ we leave only the part with delta function, as well as noise, since in the GKSL equation they are averaged by removing $\GA$ in the equation. Then we obtain the following equation for $A(t,z)$

    \begin{equation}
    	\dfrac{\pr}{\pr z}A(t,z) = i\dfrac{\be_2}{2}\dfrac{\pr^2}{\pr t^2}A(t,z) + i\ga|A(t,z)|^2A(t,z) - \dfrac{\al}{2}A(t,z).
    \end{equation}

    Going from the coordinate $z$ to time through the refractive index $n$ $z = \dfrac{c}{n}t$, where $c$ is --- the speed of light in vacuum and considering the monochromatic wave $A(t, z) = A_0 e^{- i (\omega - \omega_0) t}$, as well as dismeasuring the amplitude $A(t) = \sqrt{P_0}u(t)$, where $P_0$ is the signal power, we obtain the equation for $u(t)$

    \begin{equation}
    	\label{eq:difeq}
    	\dfrac{\pr}{\pr t} u = u\left(  -i\be_2\dfrac{\DE_{nl}^2}{2}v - \dfrac{\al}{2}v\right) + i\ga P_0v|u|^2u,
    \end{equation}
    where $\DE_{nl} = \om - \omega_0$.

    Let us compare this equation with the equation for $\braket{a}$ from \eqref{eq:FinDinOfAv} and equating the corresponding coefficients we obtain the relationship between the theoretical parameters and the experimental ones.

    \begin{equation}
    	\label{eq:parr}
    	\left\{\begin{array}{ll}
    		v\dfrac{\al}{2} = \dfrac{\ga}{4},\\
    		\be_2 v\dfrac{\DE_{nl}^2}{2} = -\DE,\\
    		\ga_{nl}P_0v = -\chi,
    	\end{array}\right.
    \end{equation}
    where $\DE_{nl}, \ga_{nl}$ --- the corresponding experimental parameters.

    However, it remains to find the relation between the external field amplitude $F$ and the experimental data. For this purpose, let us consider the characteristics of the signal. The signal has a sech form with characteristic signal time $t_p$, then

    $$
    A(t) = A_0\operatorname{sech}\left( \dfrac{t}{t_p}  \right).
    $$

    the amplitude of the external field $F$ will be proportional to the amplitude of the signal $A_0$, let us express it through the signal power.

    The total energy of the signal $E_0$ has the form

    $$
    E_0 = \int\limits_{-\infty}^{\infty}|A(t)|^2dt = A_0^2\int\limits_{-\infty}^{\infty}\operatorname{sech}^2\left(\dfrac{t}{t_p}\right)dt = 2A_0^2t_p,
    $$
    thus expressing $A_0$ through $E_0$ and $t_p$ we obtain

    \begin{equation}
    	\label{eq:parrF}
    	F = k\sqrt{\dfrac{P_0}{2}},
    \end{equation}
    where $k$ is --- the coupling parameter, $P_0 = \dfrac{E_0}{t_p}$ is --- the signal power.

    Thus, from the system \eqref{eq:parr} and the equation \eqref{eq:parrF}, knowing the experimental data, we can find the theoretical ones for our model.

\end{document}